\documentclass{iopart}

\usepackage{amsfonts}
\usepackage[english]{babel}
\usepackage{amsthm}

\newtheorem{theorem}{Theorem}
\newtheorem{lemma}{Lemma}
\newtheorem{definition}{Definition}
\newtheorem{corollary}{Corollary}
\newtheorem{remark}{Remark}

\newtheorem{proposition}{Proposition}

\newcommand{\mR}{\mathbb{R}}

\newcommand{\mN}{\mathbb{N}}
\newcommand{\mE}{\mathbb{E}}
\newcommand{\mZ}{\mathbb{Z}}
\newcommand{\mP}{\mathbb{P}}
\newcommand{\mS}{\mathbb{S}}

\newcommand{\cH}{\mathcal{H}}

\newcommand{\cF}{\mathcal{F}}
\newcommand{\cC}{\mathcal{C}}

\newcommand{\cP}{\mathcal{P}}
\newcommand{\cS}{\mathcal{S}}
\newcommand{\cR}{\mathcal{R}}

\newcommand{\ux}{\underline{x}}
\newcommand{\uxb}{\underline{x \grave{}}}

\newcommand{\pj}{\partial_{x_j}}
\newcommand{\pjb}{\partial_{{x \grave{}}_{j}}}

\newcommand{\px}{\partial_x}

\newcommand{\upx}{\partial_{\underline{x}}}
\newcommand{\upxb}{\partial_{\underline{{x \grave{}}} }}

\begin{document}
\title{Spherical harmonics and integration in superspace}
\author{H De Bie$^1$ and  F Sommen$^1$}
\address{$^1$ Clifford Research Group, Department of Mathematical Analysis, Faculty of Engineering, Ghent University, Galglaan 2, 9000 Gent, Belgium}
\ead{Hendrik.DeBie@UGent.be, fs@cage.ugent.be}
\begin{abstract}
In this paper the classical theory of spherical harmonics in $\mR^m$ is extended to superspace using techniques from Clifford analysis. After defining a super-Laplace operator and studying some basic properties of polynomial null-solutions of this operator, a new type of integration over the supersphere is introduced by exploiting the formal equivalence with an old result of Pizzetti. This integral is then used to prove orthogonality of spherical harmonics of different degree, Green-like theorems and also an extension of the important Funk-Hecke theorem to superspace. Finally, this integration over the supersphere is used to define an integral over the whole superspace and it is proven that this is equivalent with the Berezin integral, thus providing a more sound definition of the Berezin integral.
\end{abstract}
\pacs{02.30.-f,02.30.Px}
\ams{30G35, 58C50}
\submitto{\JPA}
\maketitle

\section{Introduction}

In the last decades there has been an increasing interest of mathematicians and physicists in theories studying so-called superspaces or supermanifolds. In fact, mainly two different approaches have been favored. First of all, there was the approach from algebraic geometry due to a.o. Berezin \cite{MR0208930,MR732126}, Leites \cite{MR565567} and Kostant \cite{MR0580292}, using sheaf theory to replace the structure sheaf of a differentiable manifold by a sheaf of $\mZ_2$-graded algebras.  Later developments with a perhaps more physical flavour, mainly due to DeWitt \cite{MR778559} and Rogers \cite{MR574696,MR848645}, favored an approach from differential geometry, modelling supermanifolds by glueing together spaces not anymore of the type $\mR^m$ but of the type $\mR^{p,q}_{\Lambda} = (\Lambda_0)^p \times (\Lambda_1)^q$ where $\Lambda = \Lambda_0 \oplus \Lambda_1$ is a Grassmann-algebra equipped with its natural $\mZ_2$-grading. 

Recently, we have started to explore a still different approach, namely from harmonic analysis and more specifically Clifford analysis (see \cite{MR697564,MR1169463,MR1130821}). Clifford analysis -- in its simplest form --  is a detailed study of the function theory of the Dirac operator (factorizing the Laplace-operator) in $\mR^m$ and allows for an elegant formulation of higher dimensional theories.  In fact, all dimensions are, so to speak, studied at the same time (see also further the introduction of the so-called super-dimension). The main features of Clifford analysis can be reformulated in a more abstract framework called radial algebra (see \cite{MR1472163}). This framework has several advantages; it is e.g. used in the algebraic analysis of systems of Dirac operators (see the book \cite{MR2089988} and references therein). Moreover, this framework allows us to construct a representation of Clifford analysis in superspace (see \cite{DBS1} and an earlier version in \cite{MR1771370}). This means that it is possible to introduce a set of operators on superspace such as a Dirac-operator, a Laplace-operator, Gamma-operators, etc. satisfying the axioms provided by radial algebra.

The aim of this paper is to further elaborate this new framework. First of all we will study polynomial null-solutions of the super-Laplace operator, thus generalizing the classical spherical harmonics to superspace. Then we will use an ancient formula of Pizzetti (see \cite{PIZZETTI}) to define integration of polynomials over the supersphere. This integral will turn out to have a lot of interesting properties. For example, orthogonality of spherical harmonics of different degree can immediately be proven. In this way also the so-called Funk-Hecke theorem can be generalized to superspace, providing us with the starting point for the study of super spherical integral transforms such as the spherical Fourier and Radon transforms. Finally we will use this integration over the supersphere to construct an integral over the whole superspace by generalizing the classical concept of integration in spherical co-ordinates, and we will prove that the result is equivalent to the Berezin integral.

The paper is organized as follows. We start with a short introduction to Clifford and harmonic analysis on superspace. In the following section we establish the basic properties of super-spherical harmonics, such as a Fischer-decomposition. We also determine the dimensions of spaces of spherical harmonics. Then we consider the problem of integrating polynomials over the unit-sphere in $\mR^m$ and give the solution of Pizzetti. This is used in the following section to define an integration over the so-called super-sphere. Furthermore, several properties of this integration are proven such as orthogonality of spherical harmonics and Green-like theorems. In the next section the Funk-Hecke theorem is generalized to superspace in a consistent manner. Finally the link with the Berezin integral is established.

\section{Clifford and harmonic analysis in superspace}

We consider the real algebra $\cP = Alg(x_i, e_i; {x \grave{}}_j,{e \grave{}}_j)$, $i=1,\ldots,m$, $j=1,\ldots,2n$
generated by

\begin{itemize}
\item $m$ commuting variables $x_i$ and orthogonal Clifford generators $e_i$
\item $2n$ anti-commuting variables ${x \grave{}}_i$ and symplectic Clifford generators ${e \grave{}}_i$
\end{itemize}
subject to the multiplication relations
\[ \left \{
\begin{array}{l} 
x_i x_j =  x_j x_i\\
{x \grave{}}_i {x \grave{}}_j =  - {x \grave{}}_j {x \grave{}}_i\\
x_i {x \grave{}}_j =  {x \grave{}}_j x_i\\
\end{array} \right .
\quad \mbox{and} \quad
\left \{ \begin{array}{l}
e_j e_k + e_k e_j = -2 \delta_{jk}\\
{e \grave{}}_{2j} {e \grave{}}_{2k} -{e \grave{}}_{2k} {e \grave{}}_{2j}=0\\
{e \grave{}}_{2j-1} {e \grave{}}_{2k-1} -{e \grave{}}_{2k-1} {e \grave{}}_{2j-1}=0\\
{e \grave{}}_{2j-1} {e \grave{}}_{2k} -{e \grave{}}_{2k} {e \grave{}}_{2j-1}=\delta_{jk}\\
e_j {e \grave{}}_{k} +{e \grave{}}_{k} e_j = 0\\
\end{array} \right .
\]
and where moreover all elements $e_i$, ${e \grave{}}_j$ commute with all elements $x_i$,${x \grave{}}_j$. For the motivation of these relations, we refer the reader to \cite{DBS1}, where the framework is also extended to include differential forms. The algebra $Alg(e_i; {e \grave{}}_j)$ generated by all the Clifford numbers $e_i, {e \grave{}}_j$ will be called $\cC$.

\noindent
The most important object in this algebra is the super--vector variable 
\[
x = \sum_{i=1}^m x_i e_i+ \sum_{j=1}^{2n} {x\grave{}}_j  {e \grave{}}_j.
\] 
Its square $x^2 = \sum_{j=1}^n {x\grave{}}_{2j-1} {x\grave{}}_{2j}  -  \sum_{j=1}^m x_j^2$ is scalar valued and thus generalizes the quantity $-r^2$ in $\mR^m$. The super--Dirac operator is defined as
\[
\px = \upxb-\upx = 2 \sum_{j=1}^{n} \left( {e \grave{}}_{2j} \partial_{{x\grave{}}_{2j-1}} - {e \grave{}}_{2j-1} \partial_{{x\grave{}}_{2j}}  \right)-\sum_{j=1}^m e_j \pj
\]
where $\partial_{{x \grave{}}_{i}}$ is a partial derivative with respect to an anti-commuting variable as in  e.g.\ \cite{MR565567}. A direct calculation shows that $\px x = m-2n= M$, where we encounter for the first time the so--called super--dimension $M$, taking over the r\^{o}le of the Euclidean dimension. Furthermore, when acting on $\cP$, the operators $x$ and $\px$ satisfy the following relation:
\begin{equation}
x \px + \px x  = 2 \mE + M.
\label{basicrel}
\end{equation}
Here $\mE=\sum_{j=1}^m x_j \pj +\sum_{j=1}^{2n} {x \grave{}}_j \pjb$ is the super--Euler operator, measuring the degree of homogeneous polynomials, and thus leading to the decomposition
\[
\cP = \sum_{k=0}^{\infty} \cP_k
\]
of the algebra $\cP$ in terms of the spaces $\cP_k$ of homogeneous polynomials of degree $k$.

\noindent
We can introduce a super-Laplace operator as
\[
\Delta = \px^2 =4 \sum_{j=1}^n \partial_{{x \grave{}}_{2j-1}} \partial_{{x \grave{}}_{2j}} -\sum_{j=1}^{m} \pj^2 .
\]
Finally we introduce the super-Gamma operator $\Gamma$ and the Laplace-Beltrami operator $\Delta_{LB}$ as
\begin{eqnarray*}
\Gamma &=& x \px - \mE\\
\Delta_{LB} &=& (M-2-\Gamma)\Gamma.
\end{eqnarray*}
One may easily calculate that

\[
\Delta_{LB} =x^2 \Delta - \mE(M-2 + \mE)
\]
which is completely similar to the classical expression with $M$ substituted for $m$. 

Now if we put $X = x^2/2, Y= - \Delta/2$ and $H = \mE + M/2$, then we can calculate the following commutators
\begin{eqnarray*}
[H,X] &=& 2X\\
\left[H,Y\right] &=& -2Y\\
\left[X,Y\right] &=& H
\end{eqnarray*}
proving that $X,Y$ and $H$ are the canonical generators of the Lie algebra $\mathfrak{sl}_2(\mR)$  and that we have indeed a representation of harmonic analysis in superspace (see e.g. \cite{MR1151617}). Similarly, the full Clifford analysis framework, including $x$ and $\px$, generates a representation of the superalgebra $\mathfrak{osp}(1|2)$. Note that e.g. in \cite{MR1658723}, the authors study the action of $\mathfrak{osp}(1|2)$ on the so-called $(2|2)$-dimensional supersphere in the context of DeWitt supermanifolds, including a decomposition into spherical harmonics.

In the sequel we will focus mostly on the harmonic analysis aspect and not so much on the Clifford analysis aspect, since in our opinion this probably is of more interest to an audience of physicists.

\section{Super spherical harmonics}

We begin with the following
\begin{definition}
A (super)-spherical harmonic of degree $k$ is a polynomial $H_k(x) \in \cP$ satisfying
\begin{eqnarray*}
\Delta H_k(x)&=&0\\
\mE H_k(x) &=& k H_k(x).
\end{eqnarray*}
The space of spherical harmonics of degree $k$ will be denoted by $\cH_k$.
\end{definition}

\noindent
We immediately have the following theorem concerning eigenfunctions of the Laplace-Beltrami operator:

\begin{theorem}
The space $\cH_k$ is an eigenspace of the operator $\Delta_{LB}$ corresponding to the eigenvalue $- k(M-2 + k)$.
\end{theorem}

\begin{proof}
\begin{eqnarray*}
\Delta_{LB} H_k(x) &=& x^2 \Delta H_k(x) - \mE(M-2 + \mE) H_k(x)\\
&=& - k(M-2 + k) H_k(x).
\end{eqnarray*}
\end{proof}
\noindent
We note that the result is completely similar to the classical case, upon replacing $m$ by $M$.

\begin{theorem}
The super-Laplace operator is a surjective operator on $\cP$ if and only if $m \neq 0$.
\end{theorem}

\begin{proof}
This follows immediately from the classical proof. See e.g. Lemma 3.1.2 in \cite{MR1412143}.
\end{proof}

\noindent
The previous theorem allows us to compute the dimensions (i.e. the rank as a free $\cC$-module) of the spaces $\cH_k$. Indeed we have the following 

\begin{corollary}
One has that
\[
\dim \cH_k = \dim \cP_k - \dim \cP_{k-2}
\]
where
\[
\dim  \cP_k = \sum_{i=0}^{min(k,2n)}\left( \begin{array}{c} 2n\\i \end{array}\right) \left( \begin{array}{c} k-i+m-1\\m-1 \end{array}\right)
\]
and $\dim \cP_{-1} = 0 = \dim \cP_{-2}$ by definition.
\end{corollary}

\begin{proof}
The proof is immediate, using the surjectivity of $\Delta$. The dimension of $\cP_k$ follows by a simple counting argument.
\end{proof}

\noindent
The case $m=0$ is slightly more complicated, because $\Delta$ is clearly not surjective in that case. We refer the reader to our more detailed paper \cite{DBS2}.

\noindent
Next we have the following basic formula, where $R_k$ is a homogeneous polynomial of degree $k$:
\begin{equation}
\Delta(x^2 R_k) = (4k+2M) R_k + x^2 \Delta R_k.
\label{basicequation}
\end{equation}
This formula follows from the fact that $ [ \Delta,  x^2 ] =  4 \mE + M$.

\noindent 
In the sequel we will also need the following formulae for iterated actions of the Laplace operator.

\begin{lemma}
One has the following relations:

\begin{equation*}
\begin{array}{llcl}
(i)&\Delta (x^{2t} R_{k})&=& 2t(2k+M+2t-2) x^{2t-2} R_k + x^{2t} \Delta R_k\\
(ii)&\Delta^{t+1}(x^2 R_{2t}) &=& 4(t+1)(M/2+t) \Delta^{t}( R_{2t})\\
(iii)&\Delta^{t+1}(x R_{2t+1}) &=& 2(t+1) \Delta^{t} \px( R_{2t+1}).
\end{array}
\end{equation*}
\label{relationslaplace}
\end{lemma}

\begin{proof}

The first formula is easily proven using induction on $t$ and formula (\ref{basicequation}).

\noindent
Iterating formula (\ref{basicequation}) gives

\begin{eqnarray*}
\Delta^{t+1}(x^2 R_{2t}) &=&\Delta^{t} \left( 2(4t+M) R_{2t} + x^{2} \Delta R_{2t} \right)\\
&=&  2 \left(  (4t+M) +  (4t+M-4) + \ldots +  (4+M) + M \right)\\
&& \times \Delta^{t} R_{2t}\\
&=&  2( \sum_{i=0}^t M + 4 \sum_{i=0}^t (t- i)) \Delta^{t} R_{2t}\\
&=&   2((t+1) M + 4 t(t+1) -4 t(t+1)/2) \Delta^{t} R_{2t}\\
&=&   4(t+1)(M/2+t)  \Delta^{t} R_{2t}
\end{eqnarray*}
thus proving the second statement.

Finally let us prove the last formula. We know that

\[
\Delta (x R_{2t+1}) = 2 \px R_{2t+1} + x \Delta R_{2t+1}
\]
and so
\begin{eqnarray*}
\Delta^{t+1}(x R_{2t+1}) &=& \Delta^{t} (2 \px R_{2t+1} + x \Delta R_{2t+1})\\
&=& (\sum_{i=0}^{t+1} 2) \; \px^{2t+1}( R_{2t+1}) \\
&=& 2(t+1) \px^{2t+1}( R_{2t+1}).
\end{eqnarray*}
\end{proof}

\noindent
The previous lemma leads to:

\begin{corollary}
Let $H_k \in \cH_k$ and $M \not \in -2 \mN$. Then
\begin{eqnarray*}
\Delta^i(x^{2j} H_k) &=& c_{i,j,k} x^{2j-2i} H_k, \quad i \leq j\\
&=&0, \quad i > j
\end{eqnarray*}
with
\[
c_{i,j,k} = 4^{i} \frac{j!}{(j-i)!} \frac{\Gamma(k+M/2+j)}{\Gamma(k+M/2+j-i)}.
\]
\label{laplonpieces}
\end{corollary}

\begin{proof}
By several applications of lemma \ref{relationslaplace}.
\end{proof}

Using this formula we can prove the following Fischer-decomposition of super-polynomials (for the classical case see e.g. \cite{MR0229863}):

\begin{theorem}[Fischer-decomposition]
Suppose $M \not \in -2 \mN$. Then $\cP_k$ decomposes as

\begin{equation}
\cP_k = \bigoplus_{i=0}^{\left\lfloor \frac{k}{2} \right\rfloor} x^{2i} \cH_{k-2i}.
\end{equation}
\end{theorem}

\begin{proof}
First note that the spaces $x^{2i} \cH_{k-2i}$ are all disjoint. This follows from the fact that the space $x^{2i} \cH_{k-2i}$ is the space of polynomials that are null-solutions of 
$\Delta^{i+1}$ but not of $\Delta^{i}$ (see corollary \ref{laplonpieces}). Note that this is only true if $M \not \in -2 \mN$.

\noindent
Now we calculate
\begin{eqnarray*}
\dim \bigoplus_{i=0}^{\left\lfloor \frac{k}{2} \right\rfloor} x^{2i} \cH_{k-2i} &=& \sum_{i=0}^{\left\lfloor \frac{k}{2} \right\rfloor} \dim \cH_{k-2i}\\
&=&\sum_{i=0}^{\left\lfloor \frac{k}{2} \right\rfloor} (\dim \cP_{k-2i}-\dim \cP_{k-2i-2})\\
&=& \dim \cP_k
\end{eqnarray*}
which completes the proof.
\end{proof}

\begin{remark}
It is possible to refine the previous decomposition to spherical monogenics i.e. polynomial null-solutions of the super-Dirac operator, see our paper \cite{DBS2}. 
\end{remark}

\vspace{3mm}
It is also possible to explicitly determine the Fischer-decomposition of the space $\cP_k$. This amounts to constructing projection operators $\mP_i, i=0, \ldots, \left\lfloor \frac{k}{2} \right\rfloor$ satisfying
\begin{equation}
\mP_i (x^{2j} \cH_{k-2j}) = \delta_{ij} \cH_{k-2j}.
\label{projoperators}
\end{equation}
It is immediately clear that $\mP_i$ has to be of the following form
\[
\mP_i = \sum_{j=0}^{\lfloor k/2 \rfloor -i} a_j x^{2j}\Delta^{i+j}.
\]
The coefficients $a_j$ can be determined by expressing formula (\ref{projoperators}) as a set of equations in the $a_j$. Using induction one can solve this set of equations resulting in
\[
a_j = \frac{(-1)^j}{4^{j+i} j! i!} (k-2i+M/2-1) \frac{\Gamma(k-2i-j-1+M/2)}{\Gamma(k-i+M/2)}.
\]
In the special case where $k=2t$ the projection on $\cH_0$ takes the following form
\begin{equation}
\mP_t = \frac{1}{4^t t!} \frac{\Gamma(\frac{M}{2})}{\Gamma(t+\frac{M}{2})} \Delta^t
\label{projectiononconstants}
\end{equation}
which we will need later on.

\section{Integration over the sphere in $\mR^m$: Pizzetti's formula}

Although not very well known, there exist explicit and easy formulae to calculate the integral of an arbitrary polynomial over the unit-sphere in $\mR^m$, see for example the recent papers \cite{MR1426416,MR1837866}. However, it is not obvious how to extend  in a consistent manner these formulae to superspace. Therefore we will use an old result of Pizzetti, see \cite{PIZZETTI}, expressing the integration over the sphere as an infinite sum of powers of the Laplace operator (see formula (\ref{pizformula})). As it is relevant to the sequel and moreover not easily accessible in the existing literature, we give a quick proof of this formula.

\vspace{2mm}
\noindent
We want a formula to calculate
\[
\int_{\partial B(0,1)} R d\sigma
\]
with $R$ an arbitrary polynomial, $\partial B(0,1)$ the unit-sphere in $\mR^m$ and $d \sigma$ the classical Lebesgue surface measure. 

\vspace{2mm}
\noindent
We consider two cases:

\noindent
1) $R = R_{2k}$ a homogeneous polynomial of even degree $2k$

We can calculate

\begin{eqnarray*}
\int_{\partial B(0,1)} R_{2k} d\sigma &=& \int_{\partial B(0,1)} \sum_{i=0}^{k} x^{2i} H_{2k-2i} d\sigma\\
&=& \sum_{i=0}^{k} \int_{\partial B(0,1)}  x^{2i} H_{2k-2i}d\sigma\\
&=&\sum_{i=0}^{k} (-1)^i \int_{\partial B(0,1)} H_{2k-2i}d\sigma\\
&=& (-1)^k \int_{\partial B(0,1)}   H_{0} d\sigma\\
&=& (-1)^k H_{0} \frac{2 \pi^{m/2}}{\Gamma(m/2)}.
\end{eqnarray*}
In this calculation we have used the Fischer-decomposition and the orthogonality of spherical harmonics of different degree on $\partial B(0,1)$. We can determine $H_0$ by using the projection operator (\ref{projectiononconstants}) where $M=m$. This gives 
\[
H_0 = \frac{1}{2^{2k} k!} \frac{\Gamma(m/2)}{\Gamma(k+m/2)} \Delta^{k} R_{2k}.
\]
We conclude that

\[
\int_{\partial B(0,1)} R_{2k} d\sigma = (-1)^k \frac{2 \pi^{m/2}}{2^{2k} k!\Gamma(k+m/2)} \Delta^{k} R_{2k}.
\]

\noindent
2) $R = R_{2k+1}$ a homogeneous polynomial of odd degree $2k+1$

We have that
\[
\int_{\partial B(0,1)} R_{2k+1} d\sigma=0
\]
using the same reasoning or using a symmetry argument.

Both cases can be summarized in one formula. So let $R$ be an arbitrary polynomial, then

\begin{equation}
\int_{\partial B(0,1)} R  d\sigma =  \sum_{k=0}^{\infty} (-1)^k \frac{2 \pi^{m/2}}{2^{2k} k!\Gamma(k+m/2)} (\Delta^{k} R )(0)
\label{pizformula}
\end{equation}
where the right-hand side has to be evaluated in the origin of $\mR^m$.

\begin{remark}
The formula of Pizzetti can be written in a more elegant form using Bessel functions as follows:
\[
\int_{\partial B(0,1)} R d\sigma=  2 \pi^{m/2} (P_{\frac{m}{2}-1}(\px) R) (0)
\]
with 
\[
P_{\frac{m}{2}-1}(z) = (\frac{z}{2})^{1-\frac{m}{2}} J_{\frac{m}{2}-1} (z)
\]
a kind of `normalized' Bessel function and $\px$ the Dirac-operator.
\end{remark}

\section{Integration over the supersphere}

The formula obtained in the previous section seems promising to define an integral over the formal object $x^2=-1$ which we call the supersphere. Indeed, we simply replace $m$ with $M$ to obtain 

\begin{definition}
The integral of a superpolynomial $R$ over the supersphere is given by

\begin{equation}
\int_{SS} R =  \sum_{k=0}^{\infty} (-1)^k \frac{2 \pi^{M/2}}{2^{2k} k!\Gamma(k+M/2)} (\Delta^{k} R)(0).
\label{defintss}
\end{equation}
\label{definitionintss}
\end{definition}
This integral is a linear functional which maps the space $\cP$ of super-polynomials into $\cC$.

\noindent
In the case where $M= -2t$, the first terms in the summation vanish and the formula reduces to:
\[
\int_{SS, M= -2t } R =  \sum_{k=\textbf{t+1}}^{\infty} (-1)^k \frac{2 \pi^{M/2}}{2^{2k} k!\Gamma(k+M/2)} (\Delta^{k} R)(0).
\]
Furthermore we define the area of the supersphere as
\[
\sigma^M = \int_{SS} 1 = \frac{2 \pi^{M/2}}{\Gamma(M/2)}.
\]

\vspace{3mm}
\begin{remark}
This definition has some strange consequences:
\begin{itemize}
\item in the purely fermionic case $\int_{SS} =0$, so there is no integral (see also remark \ref{otherdef})
\item in case $M= -2t-1$ the area of the supersphere can be negative (compare with the graph of the Gamma-function)
\item in case $M= -2t$ the area of the supersphere is zero.
\end{itemize}
Of course, this is not a problem as long as we do not try to interpret things in a set- or measure-theoretic way.
\end{remark}

\vspace{2mm}
\begin{remark}
One has to be very careful in the interpretation of the formulae that will be derived in the sequel. This has to do with the fact that in general functions no longer commute, which implies that

\[
\int_{SS} fg \neq \int_{SS} gf.
\] 
\end{remark}

\vspace{2mm}

The fact that functions differing a factor $x^2$ have the same integral over the supersphere, as is expected, is expressed in the following:

\begin{lemma}
One has that:
\[
\int_{SS} x^2 f = - \int_{SS} f.
\]
\end{lemma}

\begin{proof}
This follows immediately from the definition and lemma \ref{relationslaplace}.
\end{proof}

As an immediate consequence of definition \ref{definitionintss} we have the following

\begin{proposition}[Mean value]
Let $f$ be a monogenic or harmonic polynomial. Then
\[
\int_{SS} f = \frac{2 \pi^{M/2}}{\Gamma(M/2)} f(0).
\]
\end{proposition}

This integral can also be used to prove a kind of orthogonality of super spherical harmonics.

\begin{theorem}[Orthogonality]
Let $H_k, H_l$ be super spherical harmonics of degree $k,l$. If $k \neq l$ then

\[
\int_{SS} H_k H_l = 0 = \int_{SS} H_l H_k.
\]
\label{orth}
\end{theorem}

In order to prove this theorem we need a technical lemma:

\begin{lemma}
Let $H_k, H_l$ be super spherical harmonics of degree $k,l$. Then one has that
\[
\Delta(H_k H_l) = \sum H_{k-1} H_{l-1}
\]
where the right-hand side is a sum of products of spherical harmonics of degree $k-1$ and $l-1$.
\label{laplaceorth}
\end{lemma}

\begin{proof}
First note that we can split $H_k = H_k^+ + H_k^-$ where the terms in $H_k^+$ contain only even numbers of anti-commuting co-ordinates and $H_k^-$ only odd numbers. Moreover $\Delta H_k^+ = 0 = \Delta H_k^-$, because $\Delta$ is an even operator.

\noindent
We give the proof for $H_k^+ H_l$, the other part being similar.

\noindent
We have that 

\[
\partial_{x_i}^2 (H_k^+ H_l) = \partial_{x_i}^2 (H_k^+) H_l + 2 \partial_{x_i} (H_k^+) \partial_{x_i}( H_l) +  H_k^+\partial_{x_i}^2( H_l)
\]
and 

\begin{eqnarray*}
\partial_{ {x\grave{}}_{2j-1}} \partial_{ {x\grave{}}_{2j}}(H_k^+ H_l)&=&\partial_{ {x\grave{}}_{2j-1}} \partial_{ {x\grave{}}_{2j}}(H_k^+) H_l - \partial_{ {x\grave{}}_{2j}}(H_k^+)\partial_{ {x\grave{}}_{2j-1}}( H_l) \\
&&+ \partial_{ {x\grave{}}_{2j-1}}(H_k^+)\partial_{ {x\grave{}}_{2j}}( H_l) + H_k^+ \partial_{ {x\grave{}}_{2j-1}} \partial_{ {x\grave{}}_{2j}}( H_l).
\end{eqnarray*}

\noindent
So we obtain:

\begin{eqnarray*}
\Delta (H_k^+ H_l) &=& \Delta (H_k^+) H_l + H_k^+ \Delta(H_l) + \sum H_{k-1} H_{l-1}\\
&=& \sum H_{k-1} H_{l-1}\\
\end{eqnarray*}
due to the harmonicity of $H_k$ and $H_l$.
\end{proof}

We are now able to prove theorem \ref{orth}:
\begin{proof}
The case where $k+l$ is odd is trivial, because formula (\ref{defintss}) gives zero for homogeneous polynomials of odd degree. 
Now let us assume that $k+l$ is even, say $k+l=2p$. Then the integral reduces to

\[
\int_{SS} H_k H_l   =  \frac{2 \pi^{M/2}}{2^{2p} p!\Gamma(p+M/2)} \Delta^p (H_k H_l).
\]

Now we can apply Lemma \ref{laplaceorth}. Suppose e.g. that $k<l$. Then after having applied Lemma \ref{laplaceorth} $k$ times we obtain:

\[
\Delta^k (H_k H_l)= \sum H_{0} H_{l-k}.
\]
One more action of $\Delta$ gives zero, because all the factors $H_0$ are constants. So if $k \neq l$ then 

\[
\int_{SS} H_k H_l = 0  = \int_{SS} H_l H_k.
\]
\end{proof}

\begin{remark}
In the case where $M=-2t$ it is also possible to introduce another type of integration. We proceed as follows. Divide formula (\ref{defintss}) by the area of the supersphere $\sigma^M = \frac{2 \pi^{M/2}}{\Gamma(M/2)}$. Now take the limit of this expression for $
M \rightarrow -2t$. Using the fact that the Gamma-function has simple poles, this limit can be calculated for all terms in the summation with $k \leq t$ (the other terms all having limit infinity and being discarded). This yields

\[
\int_{SS, N} = \lim_{M \rightarrow -2t} \frac{1}{\sigma^M} \int_{SS} =  \sum_{k=0}^{t}  \frac{(t-k)!}{2^{2k} k! t!} (\Delta^{k} R)(0),
\]
where we have left out the terms where $k>t$.
For this new expression the integral of all polynomials of degree $\leq 2 t$ remains while vanishing for polynomials of higher degree. This definition is particularly interesting in the purely fermionic case $m=0$, $n=t$, because the previous definition would lead to the zero-operator. 
\label{otherdef}
\end{remark}

\section{The superball and Green's theorem}

In a similar way as we did for the supersphere we can formally define an integral over the superball. Again this is not a set-theoretically defined object, but an abstraction of the classical ball in $\mR^m$. We are lead to the following:

\vspace{3mm}
\begin{definition}
The integral of a superpolynomial $R$ over the superball is given by

\[
\int_{SB} R =  \sum_{k=0}^{\infty} (-1)^k \frac{\pi^{M/2}}{2^{2k} k!\Gamma(k+M/2+1)} (\Delta^{k} R)(0).
\]
\end{definition}

\vspace{2mm}
We can now easily prove the following generalization of Green's theorem.

\begin{theorem}[Green I]
Let $R$ be a superpolynomial. Then one has

\begin{equation*}
\begin{array}{llcl}
(i)&\int_{SS} x R &=& - \int_{SB} \px R\\
\vspace{-1mm}\\
(ii)&\int_{SS} \Gamma (R) &=& 0\\
\vspace{-1mm}\\
(iii)&\int_{SS} \mE R &=& -\int_{SB} \Delta R.
\end{array}
\end{equation*}
\end{theorem}

\begin{proof}
For the first expression, we only need to prove the case where $R = R_{2t+1} \in \cP_{2t+1}$. 
Then, using lemma \ref{relationslaplace},

\begin{eqnarray*}
\int_{SS} x R_{2t+1} &=& (-1)^{t+1} \frac{2 \pi^{M/2}}{2^{2t+2} (t+1)!\Gamma(t+M/2+1)} (\Delta^{t+1} x R_{2t+1})(0)\\
&=& (-1)^{t+1} \frac{2 \pi^{M/2} 2 (t+1)}{2^{2t+2} (t+1)!\Gamma(t+M/2+1)} (\Delta^{t} \px R_{2t+1})(0)\\
\end{eqnarray*}
and

\begin{eqnarray*}
\int_{SB} \px R_{2t+1} &=& (-1)^{t} \frac{ \pi^{M/2}}{2^{2t} t!\Gamma(t+M/2+1)} \Delta^{t} \px R_{2t+1}.\\
\end{eqnarray*}

\noindent
The second one is trivial, because $[\Gamma, \Delta] =0$.

\noindent
The third expression is found either by combining the previous statements using $\mE = x \px - \Gamma$, or again by direct calculation.
\end{proof}

As a consequence we immediately have the following
\begin{corollary}
Let $R$ be a superpolynomial. Then
\[
\int_{SS} \Delta_{LB} (R) = 0.
\]
\end{corollary}

We also have the following theorem, which is classically used to prove the orthogonality of spherical harmonics of different degree.

\begin{theorem}[Green II]
Let $f,g$ be two super-polynomials and let $M \not \in -2 \mN$. Then
\[
\int_{SS} (f \mE g - (\mE f) g) = -  \int_{SB} (f \Delta g - (\Delta f)  g).
\]
\end{theorem}

\begin{proof}
As $M \not \in -2 \mN$, the Fischer-decomposition exists, so it suffices to consider functions of the following form

\begin{eqnarray*}
f &=& x^{2k} H,\qquad H \in \cH_i\\
g &=& x^{2l} K,\qquad K \in \cH_j.
\end{eqnarray*}

\noindent
We distinguish two cases.

\vspace{2mm}
\noindent
1) The case where $i \neq j$. Then the left-hand side reduces to

\[
\int_{SS} (f \mE g - (\mE f) g ) = (2l+j-2k-i) \int_{SS} x^{2k+2l} H K = 0 
\]
using the orthonality of spherical harmonics of different degree.
So we need to prove that also the right-hand side is zero. We have that

\begin{eqnarray*}
f \Delta g - (\Delta f)  g &=& x^{2k} H 2l(2i+M+2l-2) x^{2l-2} K \\
&&-  2k(2j+M+2k-2)  x^{2k-2} H x^{2l} K\\
&=& \mbox{constant} \, \times  x^{2k + 2l -2} H K.
\end{eqnarray*}

\noindent
We want to calculate

\[
\int_{SB} x^{2k + 2l -2} H K .
\]
Therefore we develop $H K$ in its Fischer-decomposition (suppose $i+j=2t$, the odd case being trivial):

\[
H K = \sum_{p=0}^t x^{2p} \tilde H_{2t-2p}.
\]
All terms in this expansion give zero after integration, except for the term $x^{2t} \tilde H_{0}$.
However, using the projection operator (\ref{projectiononconstants}) we have that

\[
\tilde H_{0} = \frac{1}{4^t t!} \frac{\Gamma(\frac{M}{2})}{\Gamma(t+\frac{M}{2})} \Delta^t (H K)= 0
\] 
due to the theorem on orthogonality of spherical harmonics, which completes the proof.

\vspace{2mm}
\noindent
2) The case where $i = j$. Then the left-hand side becomes:

\begin{eqnarray*}
\int_{SS} (f \mE g - (\mE f) g) &=& (2l-2k) \int_{SS} x^{2k+2l} H K \\
& =& -(2l-2k) \int_{SS} x^{2k+2l-2} H K \\
&=&  -(2l-2k)  \frac{(-1)^{k+l+i-1} 2 \pi^{M/2}}{2^{2(k+l+i-1)} (k+l+i-1)!}\\
&& \times \frac{\Delta^{k+l+i-1}x^{2k+2l-2} H K }{\Gamma(k+l+i-1+M/2)} 
\end{eqnarray*}
and the right-hand side:

\begin{eqnarray*}
\int_{SB} (f \Delta g - (\Delta f)  g) &=& (2l-2k)(2i+M+2l+2k-2)\\
&& \times \int_{SB} x^{2k + 2l -2} H K\\
&=& (2l-2k)(2i+M+2l+2k-2) (-1)^{k+l+i-1} \\
&& \times\frac{ \pi^{M/2}}{2^{2(k+l+i-1)} (k+l+i-1)!}\\
&& \times \frac{\Delta^{k+l+i-1}x^{2k+2l-2} H K}{\Gamma(k+l+i-1+M/2+1)} 
\end{eqnarray*}
and both sides are obviously equal.
\end{proof}

\section{Funk-Hecke theorem in superspace}

Let us now introduce some numerical coefficients, which will be needed later on. We put (suppose for a moment $M \geq 2$)

\[
\alpha_l(t^k) = \sigma^{M-1} \int_{-1}^{1} t^{k} P_{l}^M(t) (1-t^2)^{\frac{M-3}{2}} dt 
\]
with 
\[
P^M_n(t) = \frac{(-1)^n}{2^n (\theta +1)( \theta +2) \ldots (\theta +n)} (1-t^2)^{-\theta} \frac{d^n}{dt^n}(1-t^2)^{\theta+n}
\]
the Legendre polynomial of degree $n$ in $M$ dimensions and $\theta = (M-3)/2$. 

Using partial integration and the definition of the Gamma function we obtain the following explicit expression for $\alpha_l(t^k)$:

\begin{eqnarray*}
\alpha_l(t^k) &=& \frac{k!}{(k-l)!} \frac{2 \pi^{\frac{M-1}{2}}}{2^l} \frac{\Gamma(\frac{k-l+1}{2})}{\Gamma(\frac{M+k+l}{2})}
  \quad \mbox{if $k+l$ even and $k \geq l$}\\
&=&0 \quad \mbox{if $k+l$ odd}\\
&=&0 \quad \mbox{if $k<l$}.\\
\end{eqnarray*}

\noindent
Note that this result is also valid for $M<2$. If $M=-2u$ ($u=0,1,2,\ldots$), we substitute the formula for $k+l$ even by

\begin{eqnarray*}
\alpha_l(t^k) &=& \frac{k!}{(k-l)!} \frac{2 \pi^{\frac{M-1}{2}}}{2^l}\frac{\Gamma(\frac{k-l+1}{2})}{\Gamma(\frac{M+k+l}{2})} \quad \mbox{if $k+l > 2u$.}\\
&=& 0 \quad \mbox{if $k+l \leq 2u$.}
\end{eqnarray*}

\noindent
Finally we also need the following coefficients, where $M= -2u$ and $k+l \leq 2u$:

\begin{eqnarray*}
\alpha_l^*(t^k) &=& \frac{k! (t-\frac{k+l}{2})!}{(k-l)! t!} \frac{ \pi^{-\frac{1}{2}} (-1)^{\frac{k+l}{2}}}{2^l} \Gamma(\frac{k-l+1}{2})
  \quad \mbox{if $k+l$ even, $k \geq l$}\\
&=&0 \quad \mbox{if $k+l$ odd}\\
&=&0 \quad \mbox{if $k<l$}.\\
\end{eqnarray*}

Now we have the following technical lemma

\begin{lemma}
The coefficients $\alpha_l(t^k)$ and $\alpha_l^*(t^k)$ satisfy the following recursion relation:

\begin{eqnarray*}
\alpha_l(t^k)&=& \frac{k}{4s(s+M/2-1)} \left( (k-1)\alpha_l(t^{k-2})  +2l \alpha_{l-1}(t^{k-1}) \right)\\
\alpha_l^*(t^k)&=& \frac{k }{4s(s + M/2-1)} \left( (k-1)\alpha_l^*(t^{k-2})  +2l \alpha_{l-1}^*(t^{k-1}) \right)\\
\end{eqnarray*}
where $2s = k+l$.
\label{recursionrel}
\end{lemma}

\begin{proof}
This follows from a careful comparison between both sides of the formulae.
\end{proof}

\noindent
Furthermore we note that $\alpha_l$ can be extended by linearity to a functional on the space of polynomial functions $f(t)$ in one variable $t$.

\noindent
Finally we put
\[
<x,y> = \frac{1}{2}(xy+yx)= - \sum_{i=1}^{m} x_i y_i +\frac{1}{2} \sum_{j=1}^{n}({x \grave{}}_{2j-1}{y \grave{}}_{2j} - {x \grave{}}_{2j} {y \grave{}}_{2j-1})
\]
where $x$ and $y$ are two independent vector variables.

\vspace{3mm}
We can now prove the following theorem (see also \cite{MR0499342} for a classical version):

\begin{theorem}[Funk-Hecke]
Let $f(t)$ be a polynomial in one variable. Let $x,y$ be independent vector variables satisying $x^2=y^2=-1$. Let $H_l$ be a spherical harmonic of degree $l$. Then
\[
\int_{SS} f(-<x,y>) H_l(x) = \alpha_l(f) H_l(y) 
\]
with $\alpha_l(f)$ as defined above.
\end{theorem}

\begin{proof}\
We first examine the case where $M \geq 2$.
It suffices to prove the theorem for $f=t^k$, $k \in \mN$. We calculate that

\begin{eqnarray*}
\Delta \left( <x,y>^k H_l \right) &=& - k(k-1) <x,y>^{k-2} H_l \\
&&+ 2k <x,y>^{k-1} \left(\sum_i y_i \partial_{x_i} + \sum_j {y \grave{}}_j \partial_{{x \grave{}}_j} \right) H_l.
\end{eqnarray*}

We now use induction on $k$. We first examine the case $k=1$, the case $k=0$ being trivial. Then using the previous relation our integral becomes (we assume $l+1=2t$, the odd case is trivial)

\begin{eqnarray*}
&&\int_{SS} (-<x,y>) H_l(x)\\
&=& - \frac{2 \pi^{M/2}}{2^{2t} t!\Gamma(t+M/2)} (-1)^t \Delta^t \left( <x,y> H_l \right)\\
&=&  \frac{2 \pi^{M/2}}{2^{2t} t!\Gamma(t+M/2)} (-1)^{t-1} \Delta^{t-1} 2 \left(\sum_i y_i \partial_{x_i} + \sum_j {y \grave{}}_j \partial_{{x \grave{}}_j} \right) H_l\\
&=& \delta_{1l} \frac{2 \pi^{M/2}}{2^{2} \Gamma(1+M/2)} 2 \left(\sum_i y_i \partial_{x_i} + \sum_j {y \grave{}}_j \partial_{{x \grave{}}_j} \right) H_l.\\
\end{eqnarray*}

\noindent
Now $H_l = H_1$ has the following general form

\[
H_1(x) = \sum_i a_i x_i + \sum_j b_j {x \grave{}}_j, \quad a_i, b_j \in \mR
\]
so

\[
\int_{SS} (-<x,y>) H_l(x) = \delta_{1l}\frac{ \pi^{M/2}}{ \Gamma(1+M/2)} H_l(y)\\
\]
which is equal to the right-hand side of the formula to be proven.

So let us now consider the induction step. Suppose the theorem holds for all $i<k$, i.e.

\[
\int_{SS} (-<x,y>)^i H_l(x) = \alpha_l(t^i) H_l(y) , \quad \mbox{for all $l$}
\]
then we prove the theorem for $t^k$. We assume that $k+l = 2s$, otherwise both sides are zero. Now we have

\begin{eqnarray*}
&&\int_{SS} (-<x,y>)^k H_l(x) \\
&=& \frac{2 \pi^{M/2}}{2^{2s} s!\Gamma(s+M/2)} (-1)^{k+s} \Delta^s \left( <x,y>^k H_l \right)\\
&=& \frac{2 \pi^{M/2}}{2^{2s} s!\Gamma(s+M/2)}  (-1)^{s-1} \Delta^{s-1} k(k-1) (-1)^{k-2} <x,y>^{k-2} H_l \\
&+&  \frac{2 \pi^{M/2}}{2^{2s} s!\Gamma(s+M/2)} (-1)^{s-1} 2k (-1)^{k-1}\\
&&\times \Delta^{s-1}  <x,y>^{k-1} \left(\sum_i y_i \partial_{x_i} + \sum_j {y \grave{}}_j \partial_{{x \grave{}}_j} \right) H_l\\
&=& \frac{k(k-1)}{4s (s+M/2-1)} \int_{SS} (-<x,y>)^{k-2} H_l(x)\\
& +& \frac{2k}{4s (s+M/2-1)} \int_{SS} (-<x,y>)^{k-1}\left(\sum_i y_i \partial_{x_i} + \sum_j {y \grave{}}_j \partial_{{x \grave{}}_j} \right) H_l\\
&=&H_l(y) \frac{k}{4s(s+M/2-1)} \left( (k-1) \alpha_l(t^{k-2}) +2l \alpha_{l-1}(t^{k-1}) \right)\\
&=& \alpha_l(t^k) H_l(y)
\end{eqnarray*}
where we have used lemma \ref{recursionrel} and the induction hypothesis.

We now discuss the other cases. If $M=-2u+1$, the proof remains the same, if we use the adapted definiton of $\alpha_l(t^k)$.
The case $M= - 2u$ is slightly more difficult, because the first $u$ terms in $\int_{SS}$ now vanish. So, one first has to prove the theorem for $\int_{SS,N}$ (see remark \ref{otherdef}) using the coefficients $\alpha_l^*(t^k)$ (which is completely similar to the above). This yields the following formula

\[
\Delta^s (-<x,y>)^k H_l(x) =  \frac{2^{2s} s! u !}{ (u-s)!} \alpha_l^*(t^k) H_l(y)
\]
with $k+l = 2s$ and $s \leq u$. We then use this formula as a first step in the induction proof for $\int_{SS}$, because the following holds

\[
\alpha_l(t^k) = \frac{(-1)^u k u! \pi^{M/2}}{2 (u+1)} \left( (k-1) \alpha_l^*(t^{k-2}) + 2l \alpha_{l-1}^*(t^{k-1}) \right)
\]
with $k+l = 2u+2$.
\end{proof}

We immediately have the following 
\begin{corollary}
Let $f(t)$ be a polynomial in one variable. Let $x,y$ satisy $x^2=y^2=-1$. Then

\[
\int_{SS} f(- <x,y>) =  \alpha_0(f(t)).
\]
\end{corollary}

\begin{remark}
Classically this result is easily obtained by realizing that a function depending on the inner product of $x$ and $y$ is constant on each hyperplane perpendicular to $y$.
\end{remark}

\vspace{3mm}
Using the previous theorem, we are able to prove the following:

\begin{corollary}[Reproducing kernel]
Let $M>1$. Then 

\[
F(x,y) = \frac{N(M,k)}{\sigma^M} P^M_k (-<x,y>)
\]
is a reproducing kernel for the space $\cH_k$, i.e.

\[
\int_{SS} F(x,y) H_l(x) = \delta_{kl} H_l(y), \quad \mbox{for all } H_l \in \cH_l,
\]
where
\[
N(M,k)= \frac{2k+M-2}{k} \left( \begin{array}{c} k+M-3\\k-1 \end{array}\right).
\]
\end{corollary}

\begin{proof}
Using the Funk-Hecke theorem we find that

\begin{eqnarray*}
&&\int_{SS} P^M_k (-<x,y>) H_l(x) \\
&=& H_l(y) \sigma^{M-1} \int_{-1}^{1} P_k^M(t) P_l^M(t) (1-t^2)^{\frac{M-3}{2}} dt.
\end{eqnarray*}
Now the orthogonality-relation of the Legendre-polynomials yields

\[
\int_{-1}^{1} P_k^M(t) P_l^M(t) (1-t^2)^{\frac{M-3}{2}} dt = \delta_{kl} \frac{\sigma^M}{\sigma^{M-1} N(M,k)}
\]
so

\[
\int_{SS} P^M_k (<x,y>) H_l(x) = \delta_{kl} H_l(y) \frac{\sigma^M}{N(M,k)}
\]
which completes the proof.
\end{proof}

The real significance of the Funk-Hecke theorem is that spaces of spherical harmonics are eigenspaces of zonal integral transformations, i.e. transformations whose kernel depends on the inner product of two generalized points on the supersphere. We consider two applications of this idea. 

\vspace{3mm}
\noindent
\textbf{Application 1: the super spherical Fourier transform.}

Consider the kernel $e^{i a t}$, where for the moment $a \in \mR$. Let us first calculate the coefficients $\alpha_l(e^{i a t})$.  We have 

\[
e^{i a t} = \sum_{k=0}^{\infty} \frac{(i a t)^k}{k!}
\]
We only need to consider the case $k \geq l$ and $k+l$ even, since the other terms are zero. So, putting $k = l+2s$, we calculate:

\begin{eqnarray*}
\alpha_l(e^{i a t}) &=& \sum_{s=0}^{\infty} \alpha_l(\frac{(i a t)^{l+2s}}{(l+2s)!})\\
&=& \sum_{s=0}^{\infty} \frac{i^{l+2s}}{(l+2s)!} \frac{(l+2s)!}{(2s)!} \frac{2 \pi^{\frac{M-1}{2}}}{2^l} \frac{\Gamma(s+1/2)}{\Gamma(M/2+l+s)} a^{l+2s}\\
&=& \sum_{s=0}^{\infty} \frac{i^{l}}{2^l}2 \pi^{\frac{M}{2}} \frac{1}{s! 2^{2s}\Gamma(M/2+l+s)} a^{l+2s}\\
&=&i^l (2\pi)^{M/2}  a^{1-M/2} J_{\frac{M}{2}+l-1}(a).
\end{eqnarray*}

\noindent
Now let us introduce the following operator ($a=1$):

\[
\cF_{SS}( . )(y) = \int_{SS} \exp(-i <x,y>)
\]
which we will call the super spherical Fourier transform. Then we obtain the following

\begin{theorem}
Let $H_l \in \cH_l$. Then

\[
\cF_{SS}( H_l(x))(y) = i^l (2\pi)^{M/2} J_{\frac{M}{2}+l-1}(1) H_l(y).
\]
\end{theorem}

\vspace{3mm}
\noindent
\textbf{Application 2: the super spherical Radon transform.}

We consider the kernel $\delta(t)$. Expansion in plane waves yields:

\[
2 \pi \delta(t) = \lim_{\nu \rightarrow \infty} \int_{-\nu}^{\nu} e^{i \xi t} d \xi.
\]
We calculate the coefficients $\alpha_l(\delta(t))$. We find, using the results on the Fourier kernel:

\begin{eqnarray*}
\alpha_l(\delta(t)) &=& \frac{1}{2 \pi} \lim_{\nu \rightarrow \infty} \int_{-\nu}^{\nu} \alpha_l(e^{i \xi t}) d \xi\\
&=&i^l (2\pi)^{M/2-1} \lim_{\nu \rightarrow \infty} \int_{-\nu}^{\nu} \xi^{1-\frac{M}{2}} J_{\frac{M}{2}+l-1}(\xi) d \xi.
\end{eqnarray*}
The remaining integral is zero if $l$ is odd. If $l$ is even, this integral can be calculated explicitly using e.g. \cite{MR0499342}, p. 245 (case $M>2$). This yields

\[
\alpha_l(\delta(t)) = 2 (-1)^{l/2} \pi^{M/2-1}\frac{\Gamma(\frac{l+1}{2})}{\Gamma(\frac{M+l-1}{2})},
\]
again completely in correspondance with the classical result (see e.g. the book \cite{MR1412143}). If we introduce the following operator:

\[
\cR_{SS}(.)(y) = \int_{SS} \delta(- <x,y>)
\]
which we will call the super spherical Radon transform, we can summarize the previous results in the following

\begin{theorem}
Let $H_l \in \cH_l$. Then if $M>2$

\begin{eqnarray*}
\cR_{SS}( H_l(x))(y) &=& 2 (-1)^{l/2} \pi^{M/2-1}\frac{\Gamma(\frac{l+1}{2})}{\Gamma(\frac{M+l-1}{2})} H_l(y), \quad \mbox{$l$ even}\\
&=& 0, \quad \mbox{$l$ odd.}
\end{eqnarray*}
\end{theorem}

\section{Integration over superspace and connection with the Berezin integral}

We can combine the previous section with the idea of integration in spherical co-ordinates in Euclidean space in order to obtain a possible definition of an integral in superspace. Indeed, the integral of a function $f$ over $\mR^m$ can be expressed as follows using spherical co-ordinates:

\[
\int_{\mR^m} f(x) dx = \int_0^{+\infty}  r^{m-1} dr \int_{\mS^{m-1}} f(r \xi) d \xi, \quad \xi \in \mS^{m-1}.
\]

It is possible to extend this recipe to superspace by substituting $m \leftrightarrow M$. If we consider a function of the following form:

\[
R_k \exp(x^2), \quad R_k \in \cP_k
\]
we then obtain the following \textit{definition} for an integral $\int_{\mR^{m |2n}}$ over the whole superspace:

\begin{eqnarray*}
\int_{\mR^{m|2n}} f &=& \int_0^{+\infty} r^{k+M-1} e^{-r^2}  dr \int_{SS} R_k\\
&=& \frac{1}{2} \Gamma(\frac{k+M}{2}) \int_{SS} R_k
\end{eqnarray*}
where the second expression is used if the integral over $r$ is divergent. Using the definition of the integral over the supersphere (see formula (\ref{defintss})) we immediately arrive at the following

\vspace{3mm}
\begin{theorem}
The integral of a function $f = R \, \exp(x^2)$ with $R$ an arbitrary super-polynomial is given by the following formula:
\begin{equation}
\int_{\mR^{m|2n}} f = \sum_{k=0}^{\infty} (-1)^k \frac{\pi^{M/2}}{2^{2k} k!} (\Delta^{k} R )(0) =\pi^{M/2}(\exp(-\Delta/4))R(0).
\label{superintegral}
\end{equation}
\end{theorem}

\vspace{3mm}
As all reference to the super-dimension $M$ has disappeared in formula (\ref{superintegral}) (except in the scaling of the formula), it seems interesting to compare this definition with the Berezin integral. This integral is defined as follows (see e.g. \cite{MR565567}). Let $f$ be an element of $C^{\infty}(\mR^m) \otimes \Lambda^{2n}$, i.e. $f$ is a superfunction with the following expansion:

\begin{equation*}
f(x, {x \grave{}}) = \sum_{\nu=(\nu_1,\ldots,\nu_{2n})} f_{\nu}(x) {x \grave{}}_{1}^{\nu_1} \ldots {x \grave{}}_{2n}^{\nu_{2n}}
\label{expansion}
\end{equation*}
where $\nu_i = 0$ or $1$ and $f_{\nu}(x)$ is a smooth function of the (real) co-ordinates $(x_1,\ldots,x_m)$. Then by definition
\[
\int_{B} f = \int_{\mR^m} f_{(1,\ldots,1)}(x) dx.
\]
In other words, we have that $\int_B =  \int_{\mR^m} dx \partial_{{x \grave{}}_{2n}} \ldots \partial_{{x \grave{}}_{1}}$.

\vspace{3mm}
Now we have the following:

\begin{theorem}
For functions $f$ of the form $R \, \exp(x^2)$ with $R$ a polynomial, the Berezin integral is equivalent with the integral defined in (\ref{superintegral}), i.e.

\begin{equation}
\int_{\mR^{m|2n}} f = \pi ^{-n} \int_{B} f.
\label{equivBerezinDBS}
\end{equation}
\label{theoremberezin} 
\end{theorem}

\begin{proof}
First note that it suffices to give the proof for functions $f$ of the following form

\begin{eqnarray*}
f &=& R_{2k} \exp{x^2}\\
&=& x_{1}^{2 \alpha_1} \ldots x_{m}^{2 \alpha_m} ({x \grave{}}_{1} {x \grave{}}_{2})^{\beta_1} \ldots ({x \grave{}}_{2n-1} {x \grave{}}_{2n})^{\beta_n} \exp(x^2), 
\end{eqnarray*}
where $\alpha_i \in \mN$, $\beta_i \in \{0,1\}$, $\sum \alpha_i + \sum \beta_i = k$ and $\sum \beta_i = l$. We will now calculate the integral of $f$ with the two definitions obtaining the same result. Let us first calculate $\int_{\mR^{m|2n}} f$. We need to calculate $\Delta^k(R_{2k})$. One immediately sees that in $\Delta^k$ only the term 

\[
(-1)^{k-l} 2^{2l} \partial_{x_{1}}^{2 \alpha_1} \ldots \partial_{x_{m}}^{2 \alpha_m} (\partial_{{x \grave{}}_{1}} \partial_{{x \grave{}}_{2}})^{\beta_1} \ldots (\partial_{{x \grave{}}_{2n-1}} \partial_{{x \grave{}}_{2n}})^{\beta_n}
\]
gives a non-zero result. As this term occurs $\frac{k !}{ \alpha_1! \ldots \alpha_m!}$ times, we obtain

\[
\Delta^k(R_{2k}) = \frac{k !}{ \alpha_1! \ldots \alpha_m!} 2^{2l} (2 \alpha_1)! \ldots (2\alpha_m)! (-1)^k.
\]

\noindent
Using this result we find that

\begin{eqnarray*}
\int_{\mR^{m|2n}} f &=&  (-1)^k \frac{\pi^{M/2}}{2^{2k} k!} (\Delta^{k} R_{2k} )(0)\\
&=& \frac{\pi^{M/2}}{2^{2k} k!} \frac{k!}{ \alpha_1! \ldots \alpha_m!} 2^{2l} (2 \alpha_1)! \ldots (2\alpha_m)!\\
&=&\pi^{-n} \frac{\pi^{m/2}}{2^{2k-2l}} \frac{(2 \alpha_1)! \ldots (2\alpha_m)!}{ \alpha_1! \ldots \alpha_m!}\\
&=&\pi^{-n} \int_{\mR^m} x_{1}^{2 \alpha_1} \ldots x_{m}^{2 \alpha_m} \exp(\ux^2) \, dx,
\end{eqnarray*}
where we have put $\ux^2 = -  \sum_{j=1}^m x_j^2$.

On the other hand let us consider the Berezin integral of $f$. We need to determine the term of $f$ in ${x \grave{}}_{1} \ldots {x \grave{}}_{2n}$. This is equivalent with determining the term of $\exp(x^2)$ in $({x \grave{}}_{1} {x \grave{}}_{2})^{1-\beta_1} \ldots ({x \grave{}}_{2n-1} {x \grave{}}_{2n})^{1-\beta_n}$. Now we calculate, putting $\uxb^2 = \sum_{j=1}^n {x\grave{}}_{2j-1} {x\grave{}}_{2j}$

\begin{eqnarray*}
\exp(x^2) &=& \exp(\ux^2 + \uxb^2)\\
&=& \exp(\ux^2)\exp(\uxb^2)\\
&=& \exp(\ux^2) (\sum_{k=0}^n \frac{\uxb^{2k}}{k!})\\
&=&  \exp(\ux^2) (({x \grave{}}_{1} {x \grave{}}_{2})^{1-\beta_1} \ldots ({x \grave{}}_{2n-1} {x \grave{}}_{2n})^{1-\beta_n} + \mbox{other terms}).\\
\end{eqnarray*}

\noindent
So 

\[
R_{2k} \exp(x^2) = x_{1}^{2 \alpha_1} \ldots x_{m}^{2 \alpha_m} \exp(\ux^2) {x \grave{}}_{1} \ldots {x \grave{}}_{2n} + l.o.t.
\]

\noindent
By definition we now have that

\[
\int_{B} f = \int_{\mR^m} x_{1}^{2 \alpha_1} \ldots x_{m}^{2 \alpha_m} \exp(\ux^2) \, dx
\]
from which the theorem follows.
\end{proof}

\begin{remark}
Note that the set of functions $f = R(x) \exp(x^2)$, considered in theorem \ref{theoremberezin}, is dense in e.g. $\cS \otimes \Lambda^{2n}$, with $\cS$ the space of rapidly decreasing functions in $\mR^m$. 
\end{remark}

In this section we thus have obtained a new way of defining integration on superspace for a sufficiently large set of functions, without resorting to the ad hoc formulation of Berezin. Note that also other attempts have been made to connect the Berezin integral with more familiar types of integration, such as given in \cite{MR784620} and \cite{MR825156} using contour integrals.

\section{Conclusions}

In this paper we have further developed a new approach to superspace. We have constructed a theory of spherical harmonics which is very similar to the classical theory in $\mR^m$. Using an old result of Pizzetti, we were able to define an integral over the supersphere, parametrised by our superdimension $M$. This integral appeared to be quite useful because it allowed us to prove a.o. orthogonality of spherical harmonics of different degree, a mean value property, Green-like theorems etc.
In this way also the important Funk-Hecke theorem could be extended to superspace. As a consequence the study of the super spherical Fourier and Radon transforms was initiated, again in complete correspondance with the classical results. 

Finally we extended our integral over the supersphere to an integral over the whole superspace, inspired by the idea of integration in spherical co-ordinates. In the resulting formula, the superdimension $M$ no longer appeared. Moreover we were able to prove that our integration recipe is equivalent to the Berezin integral. This shows that our approach is in correspondance with the classical one. It also gives a stronger motivation for certain definitions, due to the resemblance to the classical case.

\ack The first author would like to thank Henri Verschelde for a discussion leading to the results on the Berezin integral. The first author is a research assistant supported by the Fund for Scientific Research Flanders (F.W.O.-Vlaanderen).

\end{document}